\RequirePackage[l2tabu, orthodox]{nag}		
\documentclass[10pt]{article}

\usepackage[T1]{fontenc}
\usepackage{lmodern}
\usepackage[frenchmath]{mathastext}							
\usepackage{natbib}       									
\usepackage{amsmath}                        
\numberwithin{equation}{section}
\usepackage{graphicx} 											
\usepackage{enumitem}												
\usepackage{mdwlist}												
\usepackage[dvipsnames]{xcolor}							
\definecolor{Bluee}{RGB}{0,0,255}
\definecolor{Magentaa}{RGB}{255,0,255}
\usepackage[plainpages=false, pdfpagelabels]{hyperref} 
	\hypersetup{
		colorlinks   = true,
		citecolor    = Bluee, 
		linkcolor    = Magentaa, 
		urlcolor     = Turquoise
	}
\usepackage{amsmath}
\usepackage{amssymb}                        
\usepackage[mathscr]{eucal}                 
\usepackage{dsfont}													
\usepackage[paperwidth=8.5in,paperheight=11in,top=1.25in, bottom=1.25in, left=1.0in, right=1.0in]{geometry} 
\usepackage{mathtools}                      
\mathtoolsset{showonlyrefs=true}          
\usepackage{amsthm}                         
	\allowdisplaybreaks                       
	\theoremstyle{plain}
	\newtheorem{theorem}{Theorem}
	\numberwithin{theorem}{section}
	\newtheorem{proposition}[theorem]{Proposition}
	
	\theoremstyle{definition}

\renewcommand{\[}{\left[}

\newcommand\Eb{\mathds{E}}
\newcommand\Fb{\mathds{F}}

\newcommand\Pb{\mathds{P}}

\newcommand\Rb{\mathds{R}}

\newcommand\Ac{\mathscr{A}}

\newcommand\Fc{\mathscr{F}}

\newcommand\Lc{\mathscr{L}}
\newcommand\Mc{\mathscr{M}}

\newcommand\om{\omega}
\newcommand\Om{\Omega}
\newcommand\sig{\sigma}

\newcommand\gam{\gamma}
\newcommand\Gam{\Gamma}
\newcommand\lam{\lambda}
\newcommand\del{\delta}

\newcommand\Fv{\textbf{F}} 

\newcommand\Mh{\widehat{M}}

\newcommand\Mch{\widehat{\Mc}}

\newcommand\Ebt{\widetilde{\Eb}}
\newcommand\Pbt{\widetilde{\Pb}}
\newcommand\Act{\widetilde{\Ac}}

\newcommand\Pt{\widetilde{P}}
\newcommand\At{\widetilde{A}}
\newcommand\Wt{\widetilde{W}}
\newcommand\Bt{\widetilde{B}}

\newcommand\thetat{\widetilde{\theta}}
\newcommand\yt{\widetilde{y}}
\newcommand\pt{\widetilde{p}}

\newcommand\Yt{\widetilde{Y}}

\renewcommand\d{\partial}

\newcommand\ii{\mathtt{i}}
\newcommand\dd{\mathrm{d}}
\newcommand\ee{\mathrm{e}}

\newcommand{\blu}[1]{\textcolor{blue}{#1}}

\newcommand{\magenta}[1]{\textbf{\textcolor[rgb]{1,0,1}{#1}}}

\begin{document}

\title{Bond indifference prices and indifference yield curves}

\author{
Matthew Lorig
\thanks{Department of Applied Mathematics, University of Washington.  \textbf{e-mail}: \url{mlorig@uw.edu}}
}

\date{This version: \today}

\maketitle

\begin{abstract}
In a market with stochastic interest rates, we consider an investor who can either (i) invest all if his money in a savings account or (ii) purchase zero-coupon bonds and invest the remainder of his wealth in a savings account.  The indifference price of the bond is the price for which the investor could achieve the same expected utility under both scenarios.    In an affine term structure setting, under the assumption that an investor has a utility function in either exponential or power form, we show that the indifference price of a zero-coupon bond is the root of an integral expression.  As an example, we compute bond indifference prices and the corresponding indifference yield curves in the Vasicek setting and interpret the results.
\end{abstract}

%
%

\section{Introduction}
\label{sec:intro}
When the short-rate is modeled as a stochastic process, the market is incomplete and there is not a unique no-arbitrage price for a zero-coupon bond.
Utility indifference pricing (see \cite{carmona2009indifference} for an overview) provides a framework for deducing bond prices uniquely.
A number of authors have analyzed indifference prices of corporate bonds under the assumption of \textit{constant} interest rates and stochastic default intensities; see
\cite{bielecki2006indifference,sigloch2009utility,houssou2010indifference,jaimungal2012incorporating}.
By contrast, in this note, we deduce indifference prices of government bonds under the assumption of a \textit{stochastic} short-rate and no default.
Notably, unlike the vast majority of papers on indifference pricing, which present results for exponential utility only, in this note we consider both exponential and power utility.
Although government bonds are sufficiently liquid that there is no practical need to price them, our analysis is still of interest, as it provides an answer to the question: how much would a risk-averse investor be willing to pay for a bond?
\\[0.5em]
The rest of this paper proceeds as follows: in Section \ref{sec:model} we describe the dynamics of a money market account under the physical (i.e., real-world) probability measure.  We consider an investor, who wishes to maximize his expected utility of wealth at a future date.  We define the investor's indifference price for a zero-coupon bond and the corresponding indifference yield curve.  In Section \ref{sec:laplace} we compute the Laplace transform of the future value of the money market account.  Our main results appear in Section \ref{sec:indifference}, where we show that the indifference price of a bond is the root of an integral expression, assuming the investor has a utility function in either exponential form (Theorem \ref{thm:exp}) or power form (Theorem \ref{thm:pow}).  In Section \ref{sec:example} we perform a numerical study of indifference yield curves under the assumption that the short-rate has Vasicek dynamics.  Some remarks on the investor's choice of num\'eraire offered in Section \ref{sec:final-remarks}.

\section{Market model and definitions}
\label{sec:model}
Let us fix a time horizon $T < \infty$ and probability space $(\Om,\Fc,\Pb)$ with a filtration $\Fb = (\Fc_t)_{0 \geq t \leq T}$.  The measure $\Pb$ represents the physical (i.e., real-world) probability measure and the filtration $\Fb$ represents the history of the market.  Suppose an investor in this market can invest in two assets (i) a money market account, which grows at the \textit{short-rate} $R = (R_t)_{0 \leq t \leq T}$, and (ii) a \textit{$T$-maturity zero-coupon bond} $P^T=(P_t^T)_{0 \leq t \leq T}$, which pays one unit of currency at the maturity date $T$.  We wish to answer the following question: \textit{How much would this investor be willing to pay for $\nu$ zero-coupon bonds}?
\\[0.5em]
To answer this question, we suppose that the dynamics of the \textit{short-rate} $R = (R_t)_{0 \leq t \leq T}$ are of the form
\begin{align}
\dd R_t 
	&=	\mu(t,R_t) \dd t + \sig(t,R_t) \dd W_t ,	\label{eq:dR}
\end{align}
where $W$ is a $(\Pb,\Fb)$-Brownian motion.   We shall assume the $\mu$ and $\sig$ are such that the short rate belongs to the class of one-factor affine term-structure (ATS) models
\begin{align}
\mu(t,r)
	&=	\mu_0(t) + \mu_1(t) r, &
\sig^2(t,r)
	&=	\sig_0(t) + \sig_1(t) r . \label{eq:affine}
\end{align}
See \cite[Chapter 5]{filipovic2009term} for an overview.
Let $X = (X_t)_{0 \leq t \leq T}$ denote the value of a portfolio invested entirely in the money market account.  The dynamics of $X$ are
\begin{align}
\dd X_t
	&=	R_t X_t \dd t . 	\label{eq:dX}
\end{align}
Consider now, an investor who, at time $t$, has $x$ units of currency invested in the money market account and also owns $\nu$ zero-coupon bonds.
The investor's expected utility at time $T$ is
\begin{align}
V(t,x,r;T,\nu,\gam)
	&:=	\Eb_{t,x,r} U\Big( \frac{X_T + \nu P_T^T}{P_T^T} ; \gam \Big) 
	= 	\Eb_{t,x,r} U( X_T + \nu ; \gam ) ,
\end{align}
where $U$ is the investor's \textit{utility function} and $\gam$ is the investor's \textit{risk aversion}.  Observe that the investor is using the bond $P^T$ as num\'eraire.  This is a logical choice as, once purchased, the bond provides a known rate of return over the interval $[t,T]$, whereas the return of the money market account over the interval $[t,T]$ is stochastic.  For the above investor, the money market account is a riskier investment than the bond.
\\[0.5em]
Now, consider an investor who has two investment options: he can either
(i) invest all of his money in the money market account, or
(ii) purchase $\nu$ bonds for $p$ units of currency each, and invest the rest of his wealth in the money market account.
Assuming a time $t$ total wealth $x$, both options would yield the same expected utility if the price $p$ per bond satisfied
 \begin{align}
V(t,x,r;T,0,\gam)
	&=	V(t,x - \nu p,r;T,\nu,\gam) . \label{eq:p-def}
\end{align}
Thus, we define the investor's \textit{indifference price} for $\nu$ zero-coupon bonds as the solution $p \equiv p(t,x,r;T,\nu,\gam)$ of \eqref{eq:p-def}.
If the market price of a bond were higher than the indifference price, then the investor would prefer to invest his entire wealth in the money market account.
If the market price of a bond were lower than the indifference price, then the investor would prefer to purchase $\nu$ bonds and invest the rest of his wealth in the money market account.
\\[0.5em]
Frequently, bond prices are described in terms of their \textit{yield} $Y^T = (Y_t^T)_{0 \leq t \leq T}$, where
\begin{align}
Y_t^T	
	&:=	\frac{ -  1 }{T-t} \log P_t^T .
\end{align}
Accordingly, we define the investor's \textit{indifference yield} $y(t,x,r;T,\nu,\gam)$ as follows
\begin{align}
y(t,x,r;T,\nu,\gam)
	&:=	\frac{-1}{T-t} \log p(t,x,r;T,\nu,\gam) . \label{eq:y-def}
\end{align}
If the market yield of a bond were lower than the indifference yield, then the investor would prefer to invest his entire wealth in the money market account.
If the market yield of a bond were higher than the indifference yield, then the investor would prefer purchase $\nu$ bonds an invest the rest of his wealth in the money market account.

\section{Laplace transform of the money market account $X_T$}
\label{sec:laplace}

In order to deduce the indifference price of a bond, we shall need the following proposition.

\begin{proposition}
Let $L$ denote the \textit{Laplace transform} of $X_T$
\begin{align}
L(t,x,r;T,z)
	&:= \Eb_{t,x,r} \ee^{- z X_T} . \label{eq:L-def}
\end{align}
Then the function $L$ satisfies
\begin{align}
L(t,x,r;T,z)
	&=	\frac{1}{2\pi} \int_\Rb \dd \om_r \, z^{\ii \omega } \Gam( - \ii \omega ) \ee^{A(t,\om;T) + r B(t,\om;T) + \ii \om \log x} , & 
\om_i
	> 0 , \label{eq:L}
\end{align}
where $\Gamma$ denotes a Gamma function, $\om = \om_r + \ii \om_i$ and $A$ and $B$ satisfy
\begin{align}
0
	&=	\d_t A + \mu_0 B + \tfrac{1}{2} \sig_0 B^2 , &
A(T,\om;T)
	&=	0 , \label{eq:A-ode} \\
0
	&=	\d_t B + \mu_1 B + \tfrac{1}{2} \sig_1  B^2 + \ii \om , &
B(T,\om;T)
	&=	0 . \label{eq:B-ode}
\end{align}
\end{proposition}

\begin{proof}
The function $L$ satisfies the following partial differential equation (PDE)
\begin{align}
(\d_t + \Lc) L
	&= 0 , &
L(T,x,r;T,z)
	&=	\ee^{- z x} ,
\end{align}
where $\Lc$ is the generator of $(X,R)$ and is given explicitly by
\begin{align}
\Lc
	&=	r x \d_x + \mu(t,r) \d_r + \tfrac{1}{2} \sig^2(t,r) \d_r^2 .
\end{align}
Consider the following change of variables
\begin{align}
L(t,x,r;T,z)
	&=	M(t,q(x),r;T,z) , &
q(x)
	&=	\log x . \label{eq:var-change}
\end{align}
It is easy to show that the function $M$ satisfies
\begin{align}
(\d_t + \Mc ) M
	&=	0 , &
M(T,q,r;T,z)
	&=	\exp( - z \ee^q ) ,
\end{align}
where the operator $\Mc$ is given by
\begin{align}
\Mc
	&=	 r \d_q + \mu(t,r) \d_r + \tfrac{1}{2} \sig^2(t,r) \d_r^2 .
\end{align}
Now, let $\Mh$ denote the generalized Fourier transform of $M$ with respect to the $q$ variable
\begin{align}
\Mh(t,\om,r;T,z)
	&:= \Fv[ M(t,\cdot,r;T,z) ](\om)
	=		\int_\Rb \dd q \,  M(t,q,r;T,z) \ee^{- \ii \om q} . 
\end{align}
Then, for any $\om = \om_r + \ii \om_i$ such that $\om_i > 0$, the function $\Mh$ satisfies
\begin{align}
(\d_t + \Mch)\Mh
	&=	0 , &
\Mh(T,\om,r;T,z)
	&= z^{\ii \omega } \Gam( - \ii \omega ) , \label{eq:M-hat-pde}
\end{align}
where the operator $\Mch$ is given by
\begin{align}
\Mch
	&=	\ii \om r + \mu(t,r) \d_r + \tfrac{1}{2} \sig^2(t,r) \d_r^2 .
\end{align}
Now, consider the following Ansatz
\begin{align}
\Mh(t,\om,r;T,z)
	&=	z^{\ii \omega } \Gam( - \ii \omega ) \ee^{A(t,\om;T) + r B(t,\om;T)} , &
A(T,\om;T)
	&=	0 , &
B(T,\om;T)
	&=	 0 . \label{eq:ansatz}
\end{align}
Observe that the Ansatz $\Mh$ in \eqref{eq:ansatz} satisfies the terminal condition in \eqref{eq:M-hat-pde}.  Inserting the Ansatz \eqref{eq:ansatz} into the PDE \eqref{eq:M-hat-pde}, recalling that $\mu$ are $\sig$ are of the form \eqref{eq:affine} and collecting terms of like order in $r$ we find that $A$ and $B$ satisfy the system of coupled ordinary differential equations (ODEs) given by \eqref{eq:A-ode} and \eqref{eq:B-ode}.
Assuming $A$ and $B$ can be computed either analytically or numerically, then the function $M$ is obtained by taking the inverse Fourier transform of $\Mh$.  We have
\begin{align}
M(t,q,r;T,z)
	&=	\Fv^{-1}[ \Mh(t,\cdot,r;T,z) ](q)
	=		\frac{1}{2\pi} \int_\Rb \dd \om_r \, \Mh(t,\om,r;T,z) \ee^{ \ii \om q} \\
	&=	\frac{1}{2\pi} \int_\Rb \dd \om_r \, z^{\ii \omega } \Gam( - \ii \omega ) \ee^{A(t,\om;T) + r B(t,\om;T) + \ii \om q} .
\end{align}
Finally, undoing the variable change \eqref{eq:var-change}, we find that the Laplace transform of $X_T$ is given by \eqref{eq:L}.
\end{proof}

\section{Bond indifference prices and indifference yields}
\label{sec:indifference}

Having obtained the Laplace transform of the money market account $X_T$, we are now in a position to show that the indifference price of a bond is the root of an integral expression.  We will consider two cases (i) the investor has a utility function of exponential form, and (ii) the investor has a utility function in the power form.

\begin{theorem}[Exponential utility]
\label{thm:exp}
Suppose the investor's utility function $U$ is of the exponential form
\begin{align}
U(x;\gam)
	&=	\frac{-1}{\gam} \ee^{- \gam x} , &
\gam
	&> 0 , \label{eq:U-exp}
\end{align}
Then the indifference price $p \equiv p(t,x,r;T,\nu, \gam)$ satisfies
\begin{align}
0
	&=	 \int_\Rb \dd \om_r  \, \gam^{\ii \omega } \Gam( - \ii \omega ) \ee^{A(t,\om;T) + r B(t,\om;T)} 
				\Big( \ee^{\ii \om \log x} - \ee^{\ii \om \log (x - \nu p) - \gam \nu} \Big) , &
\om_i
	&>	0 , \label{eq:p-exp}
\end{align}
and the indifference yield $y \equiv y(t,x,r;T,\nu, \gam)$ satisfies
\begin{align}
0
	&=	 \int_\Rb \dd \om_r  \, \gam^{\ii \omega } \Gam( - \ii \omega ) \ee^{A(t,\om;T) + r B(t,\om;T)} 
				\Big( \ee^{\ii \om \log x} - \ee^{\ii \om \log (x - \nu \ee^{- (T-t)y}) - \gam \nu} \Big) , &
\om_i
	&>	0 , \label{eq:y-exp}
\end{align}
where $A$ and $B$ are the solutions of \eqref{eq:A-ode} and \eqref{eq:B-ode}, respectively
\end{theorem}

\begin{proof}
Using \eqref{eq:L-def}, \eqref{eq:L} and \eqref{eq:U-exp} we have
\begin{align}
V(t,x,r;T,\nu, \gam)
	&=	\frac{-1}{\gam} \Eb_{t,x,r} \ee^{- \gam (X_T + \nu, \gam)}
	=		\frac{-\ee^{-\gam \nu}}{\gam} L(t,x,r;T,\gam) \\
	&=	\frac{-\ee^{-\gam \nu}}{2 \pi \gam} \int_\Rb \dd \om_r \, \gam^{\ii \omega } \Gam( - \ii \omega ) \ee^{A(t,\om;T) + r B(t,\om;T) + \ii \om \log x} . \label{eq:V-exp}
\end{align}
From \eqref{eq:p-def} 
 we find that the indifference price $p \equiv p(t,x,r;T,\nu, \gam)$ satisfies
\begin{align}
0
	&=	- 2 \pi \gam \Big( V(t,x,r;T,0,\gam) - V(t,x-\nu p,r;T,\nu, \gam) \Big)  . \label{eq:no-name-exp}
\end{align}
Inserting \eqref{eq:V-exp} into \eqref{eq:no-name-exp} and simplifying yields \eqref{eq:p-exp}.
The expression \eqref{eq:y-exp} for the indifference yield  $y \equiv y(t,x,r;T,\nu, \gam)$ follows from \eqref{eq:y-def} and \eqref{eq:p-exp}.
\end{proof}

\begin{theorem}[Power utility]
\label{thm:pow}
Suppose the investor's utility function $U$ is of the power form
\begin{align}
U(x;\gam)
	&=	\frac{x^{1-\gam}}{1-\gam} , &
\gam
	&\in (0,1) \cup (1,\infty) . \label{eq:U-pow} 
\end{align}
Then the indifference price $p$ satisfies
\begin{align}
0
	&=	\int_0^\infty \frac{\dd z}{z^{2-\gam}} \int_\Rb \dd \om_r \, 
			z^{\ii \omega } \Gam( - \ii \omega ) \ee^{A(t,\om;T) + r B(t,\om;T)} \Big( \ee^{\ii \om \log x} - \ee^{\ii \om \log (x - \nu p) - z \nu} \Big) , &
\om_i
	&>	0, \label{eq:p-pow}
\end{align}
and the indifference yield $y \equiv y(t,x,r;T,\nu, \gam)$ satisfies
\begin{align}
0
	&=	\int_0^\infty \frac{\dd z}{z^{2-\gam}} \int_\Rb \dd \om_r \, 
			z^{\ii \omega } \Gam( - \ii \omega ) \ee^{A(t,\om;T) + r B(t,\om;T)} \Big( \ee^{\ii \om \log x} - \ee^{\ii \om \log (x - \nu \ee^{-(T-t)y}) - z \nu} \Big) , &
\om_i
	&>	0 , 			\label{eq:y-pow}
\end{align}
where $A$ and $B$ are the solutions of \eqref{eq:A-ode} and \eqref{eq:B-ode}, respectively.
\end{theorem}

\begin{proof}
In the case that $\gam \in (0,1)$, using the following identity from from \cite{schurger2002laplace}
\begin{align}
x^\alpha
	&=	\frac{\alpha}{\Gam(1-\alpha)} \int_0^\infty \frac{\dd z }{z^{1+\alpha}} \Big( 1 - \ee^{- z x} \Big)  , &
\alpha
	&\in (0,1) ,
\end{align}
we obtain from \eqref{eq:L-def}, \eqref{eq:L} and \eqref{eq:U-pow} that
\begin{align}
V(t,x,r;T,\nu, \gam)
	&=	\frac{1}{\Gam(\gam)} \int_0^\infty \frac{\dd z}{z^{2-\gam}} \Big( 1 - \Eb_{t,x,r}\ee^{- z (X_T + \nu, \gam)} \Big)  \\
	&=	\frac{1}{\Gam(\gam)} \int_0^\infty \frac{\dd z}{z^{2-\gam}} \Big( 1 - \ee^{- z \nu} L(t,x,r;T,z) \Big) \\
	&=	\frac{1}{\Gam(\gam)} \int_0^\infty \frac{\dd z}{z^{2-\gam}} 
			\Big( 1 - \frac{\ee^{- z \nu}}{2 \pi} \int_\Rb \dd \om_r \, z^{\ii \omega } \Gam( - \ii \omega ) \ee^{A(t,\om;T) + r B(t,\om;T) + \ii \om \log x} \Big) . \label{eq:V-frac-pow}
\end{align}
And, in the case that $\gam \in (1,\infty)$, using the following identity from from \cite{schurger2002laplace}
\begin{align}
\frac{1}{x^\alpha}
	&=	\frac{1}{\Gam(\alpha)} \int_0^\infty \dd z \, z^{\alpha - 1} \ee^{- z x} , &
\alpha
	&> 0 ,
\end{align}
we obtain from \eqref{eq:L-def}, \eqref{eq:L} and \eqref{eq:U-pow} that
\begin{align}
V(t,x,r;T,\nu, \gam)
	&=	\frac{1}{(1-\gam)\Gam(\gam-1)} \int_0^\infty \frac{\dd z}{z^{2-\gam}} \Eb_{t,x,r} \ee^{- z (X_T + \nu, \gam)} \\
	&=	\frac{-1}{\Gam(\gam)} \int_0^\infty \frac{\dd z}{z^{2-\gam}} \ee^{- \nu z} L(t,x,r;T,z) \\
	&=	\frac{-1}{\Gam(\gam)} \int_0^\infty \frac{\dd z}{z^{2-\gam}}
			\frac{\ee^{- z \nu}}{2 \pi} \int_\Rb \dd \om_r \, z^{\ii \omega } \Gam( - \ii \omega ) \ee^{A(t,\om;T) + r B(t,\om;T) + \ii \om \log x} . \label{eq:V-neg-pow}
\end{align}
From \eqref{eq:p-def} we have
\begin{align}
0
	&=	- 2 \pi \Gam(\gam) \Big( V(t,x,r;T,0,\gam) - V(t,x-\nu p,r;T,\nu, \gam) \Big) . \label{eq:no-name-pow}
\end{align}
Inserting either \eqref{eq:V-frac-pow} or \eqref{eq:V-neg-pow} into \eqref{eq:no-name-pow} we find that, for both $\gam \in (0,1)$ and $\gam \in (1,\infty)$, the indifference price $p$ satisfies \eqref{eq:p-pow}.
The expression \eqref{eq:y-pow} for the indifference yield $y \equiv y(t,x,r;T,\nu, \gam)$ follows from \eqref{eq:y-def} and \eqref{eq:p-pow}.
\end{proof}


\noindent
Indifference bond prices and yields can be easily found using a numerical root-finding algorithm such as \texttt{FindRoot[]} in Wolfram's Mathematica.

\section{Example: Vasicek model}
\label{sec:example}
In this section, we assume that the short-rate $R$ has Vasicek dynamics
\cite{VASICEK1977177}
\begin{align}
\dd R_t
	&=	\kappa ( \theta - R_t) \dd t + \del \dd W_t . \label{eq:dR-vasicek}
\end{align}
Note that the Vasicek model is an ATS model with
\begin{align}
\mu_0(t)
	&=	\kappa \theta , &
\mu_1(t)
	&=	- \kappa , &
\sig_0(t)
	&=	\del^2 , &
\sig_1(t)
	&=	0 . \label{eq:mu-sig-vasicek}
\end{align}
Inserting \eqref{eq:mu-sig-vasicek} into \eqref{eq:A-ode} and \eqref{eq:B-ode} and solving the ODEs we obtain
\begin{align}
A(t,\om;T)
	&= \frac{\ii \om \theta}{\kappa} \Big( \ee^{-\kappa (T-t)} - 1 + \kappa ( T - t) \Big)
			+ \frac{\om^2 \del^2}{4 \kappa^3} \Big( \ee^{-2\kappa (T-t)} - 4 \ee^{-\kappa (T-t)} + \big( 3 - 2 \kappa ( T-t) \big) \Big), \\
B(t,\om,T)
	&=	\frac{\ii \om}{\kappa} \Big( 1 - \ee^{-\kappa (T-t)} \Big) .
\end{align}
Now, suppose that, with the money market account as num\'eraire, the market were pricing bonds under a \textit{risk-neutral (i.e., martingale) measure} $\Pbt^{(\lam)}$, which is defined by a constant \textit{market price of interest rate risk} $\lam$
\begin{align}
\frac{\dd \Pbt^{(\lam)}}{\dd \Pb}
	&=	\exp \Big( - \frac{1}{2} \lam^2 T - \lam W_T \Big) .
\end{align}
Then the dynamics of $R$ under $\Pbt^{(\lam)}$ would be
\begin{align}
\dd R_t
	&=	\kappa ( \thetat(\lam) - R_t) \dd t + \del \dd \Wt_t^{(\lam)} , &
\thetat(\lam)
	&:=	\theta - \del \lam / \kappa ,
\end{align}
where the process $\Wt^{(\lam)} = (\Wt_t^{(\lam)})_{0 \leq t \leq T}$, defined by $\Wt_t^{(\lam)} = W_t + \lam t$, is a $(\Pbt^{(\lam)},\Fb)$-Brownian motion.
\textit{market bond prices} $\Pt^{T,(\lam)} = (\Pt_t^{T,(\lam)})_{0 \leq t \leq T}$ and the corresponding \textit{market yields} $\Yt^{T,(\lam)} = (\Yt_t^{T,(\lam)})_{0 \leq t \leq T}$ would be given by
\begin{align}
\Pt_t^{T,(\lam)}
	&=	\Ebt^{(\lam)} \Big( \ee^{- \int_t^T R_s \dd s} \Big| \Fc_t \Big)
	=: \pt(t,R_t;T,\lam) , \\
\Yt_t^{T,(\lam)}
	&=	\frac{- \log \Pt_t^{T,(\lam)} }{T-t}
	=		\frac{- \log \pt(t,R_t;T,\lam) }{T-t}
	=:	\yt(t,R_t;T,\lam) .
\end{align}
The function $\pt$ satisfies the following PDE
\begin{align}
(\d_t - r + \Act^{(\lam)}) \pt 
	&=	0 , &
\pt(T,r;T,\lam)
	&=	1 ,
\end{align}
where the operator $\Act^{(\lam)}$ is the generator of $R$ under $\Pbt^{(\lam)}$ and is given explicitly by
\begin{align}
\Act^{(\lam)}
	&=	\kappa (\thetat(\lam) - r) \d_r + \tfrac{1}{2}\del^2 \d_r^2 .
\end{align}
As $R$ has ATS dynamics under $\Pbt^{(\lam)}$, one can easily verify that $\pt$ is given by
\begin{align}
\pt(t,r;T,\lam)
	&=	\ee^{\At(t;T,\lam) + r \Bt(t;T,\lam)} , \label{eq:pt} 
\end{align}
where the functions $\At$ and $\Bt$ are defined as follows
\begin{align}
\At(t;T,\lam)
	&:=	A(t,\ii;T), &
\Bt(t;T,\lam)
	&:=	B(t,\ii;T) , &
\theta
	&\to \thetat(\lam) . \label{eq:thetat}
\end{align}
Now, let us consider an investor with an exponential utility function, as described in Theorem \ref{thm:exp}.  Fix the following parameters
\begin{align}
\kappa
	&=	0.05 , &
\theta
	&=	0.03 , &
\del
	&=	\sqrt{2 \kappa^2 \theta} , & 
t
	&=	0, &
x
	&=	5, &
r
	&=	0.01 . \label{eq:parameters}
\end{align}
On the left side of Figure \ref{fig:gam} we plot the indifference yield curve $y(t,x,r;T,\nu, \gam)$ (i.e., the solution of \eqref{eq:y-exp}) as a function of $T$ with the number of bonds purchased fixed at $\nu = 1$ and with the risk-aversion parameter varying from $\gam = 0.1$ (\blu{blue}) to $\gam = 0.2$ (\magenta{magenta}).  For comparison, we also plot the market yield curve $\yt(t,x,r;T,\lam)$ (dashed, black) assuming the market prices of risk is zero $\lam = 0$.  Observe that as the risk-aversion parameter $\gam$ increases the corresponding indifference yield curves decrease.  This is consistent with the intuition that, the more risk-averse an investor is, the more willing he would be to accept a low but known rate of return of a bond rather than a potentially higher but uncertain rate of return of a money market account.
\\[0.5em]
On the right side of Figure \ref{fig:gam} we plot the indifference yield curve $y(t,x,r;T,\nu, \gam)$ as a function of $T$ with the risk-aversion parameter $\gam = 0.15$ fixed and with the number of bonds purchased varying from $\nu = -4.5$ (\blu{blue}) to $\nu = 4.5$ (\magenta{magenta}).  Note that a negative $\nu$ indicates selling bonds.  For comparison, we also plot the market yield curve $\yt(t,x,r;T,\lam)$ (dashed, black) assuming the market prices of risk is zero $\lam = 0$.  Observe that as the number of bonds purchased increases the corresponding indifference yield curves also increase.  This is consistent with the intuition that a risk-averse seller of a bond will ask a higher price than a risk-averse buyer of a bond is willing to bid (thereby resulting in a bid-ask spread).
\\[0.5em]
Using \eqref{eq:pt}, a direct computation shows that the dynamics of the market bond price $\Pt^{T,(\lam)}$ under the physical probability measure $\Pb$ satisfies
\begin{align}
\dd \Big( \frac{\Pt_t^{T,(\lam)} }{X_t} \Big)
	&=	Q_t^{T,(\lam)} \Big( \frac{\Pt_t^{T,(\lam)} }{X_t} \Big) \dd t + S_t^{T,(\lam)} \Big( \frac{\Pt_t^{T,(\lam)} }{X_t} \Big) \dd W_t ,
\end{align}
where the drift $Q^{T,(\lam)}$ and volatility $S^{T,(\lam)}$ are given by
\begin{align}
Q_t^{T,(\lam)}
	&= \frac{1}{\pt(t,R_T;T,\lam)} \Big( \d_t - R_t + \kappa( \theta - R_t) \d_r + \tfrac{1}{2} \del^2 \d_r^2 \Big) \pt(t,R_T;T,\lam) , \label{eq:Q} \\
S_t^{T,(\lam)}
	&=	\frac{1}{\pt(t,R_T;T,\lam)}\del \d_r \pt(t,R_T;T,\lam) . \label{eq:S}
\end{align}
One can show from \eqref{eq:pt} that
\begin{align}
\frac{Q_t^{T,(\lam)}}{ S_t^{T,(\lam)}}
	&=	\lam ,
\end{align}
which allows us to interpret the market price of risk $\lam$ as the return $Q^{T,(\lam)}$ that the market pays an investor in order to bear the risk of holding a bond with volatility $S^{T,(\lam)}$.  Using the above, we can (and we do) define the \textit{investor's indifference price of interest rate risk} as the solution $\lam \equiv \lam(t,x,r;T,\nu, \gam)$ of
\begin{align}
p(t,x,r;T,\nu, \gam)
	&=	\pt(t,r;T,\lam) .
\end{align}
On the left side of Figure \ref{fig:lambda}, we plot the market yield $\yt(t,r;T,\lam)$ as a function of $T$ with the market price of risk $\lam$ varying from $\lam = -0.05$ (\blu{blue}) to $\lam = 0.05$ (\magenta{magenta}).  For reference, we also plot $\yt(t,r;T,0)$ (dashed black).  In the figure, we see that, as the market price of risk $\lam$ increases, the corresponding yield curve decreases.  This is to be expected, as the long-run mean $\thetat(\lam)$ of $R$ under $\Pbt^{(\lam)}$ is a decreasing function of $\lam$.  A lower long-run mean $\thetat(\lam)$ results in a higher bond price and therefore a lower yield.
\\[0.5em]
On the right-hand side of Figure \ref{fig:lambda} we plot the investor's indifference price of interest rate risk $\lam(t,x,r;T,\nu,\gam)$ as a function of $T$ with the investor's risk aversion varying from $\gam = 0.1$ (\blu{blue}) to $\gam = 0.2$ (\magenta{magenta}).  To understand the plot, consider the following: 
the investor described in Section \ref{sec:model} has chosen the bond as num\'eraire.  For this investor, it is the money market account, rather than the bond, that is the risky asset.  A routine application of It\^o's Lemma shows that $X/P^{T,(\lam)}$ satisfies
\begin{align}
\dd \Big( \frac{X_t}{P_t^{T,(\lam)}} \Big)
	&=	\Big( (S_t^{T,(\lam)})^2 - Q_t^{T,(\lam)} \Big) \Big( \frac{X_t}{P_t^{T,(\lam)}} \Big) \dd t + S_t^{T,(\lam)} \Big( \frac{X_t}{P_t^{T,(\lam)}} \Big) ( - \dd W_t  ).
\end{align}
One can show from \eqref{eq:pt}, \eqref{eq:Q} and \eqref{eq:S} that
\begin{align}
\frac{(S_t^{T,(\lam)})^2 - Q_t^{T,(\lam)}}{ S_t^{T,(\lam)}}
	&=	- \lam - \frac{\del}{\kappa} \Big( 1 - \ee^{-\kappa(T-t)} \Big) .
\end{align}
The higher investor's indifference price of interest rate risk $\lam$, the lower the reward the investor perceives for bearing the risk of holding the money market account.  Thus, as the investor's risk-aversion $\gam$ goes up, we expect the investor's implied price of interest rate risk to go up as well, and this is precisely what we see in Figure \ref{fig:lambda}.

\section{On the investor's choice of num\'eraire}
\label{sec:final-remarks}
Although we have not focused on it, one could have repeated the analysis of this manuscript assuming the investor had used a money market account $M = (M_t)_{0 \leq t \leq T}$, rather than the bond, as his num\'eraire.  In this case, if the investor at $x$ units of currency invested in the moneny market account at time $t$ and additionally owned $\nu$ bonds, his expected utility would be
\begin{align}
V(t,x,m,r;T,\nu,\gam)
	&:=	\Eb_{t,x,m,r} U \Big( \frac{X_T + \nu P_T^T}{M_T}; \gam \Big) 
	=		\Eb_{t,x,m,r} U \Big( \frac{X_T + \nu }{M_T}; \gam \Big) ,
\end{align}
where $\dd M_t = R_t M_t \dd t$.  In this scenario, the bond $P^T$ would be considered the risky asset and the money market account $M$ would be considered the safe asset.  Rather than repeat the entire analysis here, we simply note that, if the investor has exponential utility, the indifference prices of a bond can be computed explicitly as a numerical integral
\begin{align}
p(t,m,r;T,\nu,\gam)
	&=	- \frac{m}{\gam \nu } \log \frac{1}{2\pi}\int_\Rb \dd \om_r \, \Gamma ( \ii \omega ) \ee^{A(t,\om;T) + r B(t,\om;T) + \ii \om \log \frac{m}{\gam \nu}}  , &
\om_i
	&< 0 ,
\end{align}
and if the investor has power utility, the indifference price  of a bond $p \equiv p(t,x,m,r;T,\nu,\gam)$ satisfies
\begin{align}
0
	&=	\int_0^\infty \frac{\dd z}{z^{2-\gam}} \ee^{- z x/ m} \Big(
			1 - \frac{\ee^{ z \nu p / m}}{2 \pi} \int_\Rb \dd \om_r \, \Gam( \ii \omega ) \ee^{A(t,\om;T) + r B(t,\om;T) + \ii \om \log \frac{m}{z \nu}}
			\Big)  , &
\om_i
	&< 0 ,
\end{align}
where $A$ and $B$ are the solutions to \eqref{eq:A-ode} and \eqref{eq:B-ode}, respectively.

%
%

\clearpage

\bibliographystyle{chicago}
\bibliography{references}	

%
%

\clearpage

\begin{figure}
\begin{tabular}{cc}
\includegraphics[width=0.48\textwidth]{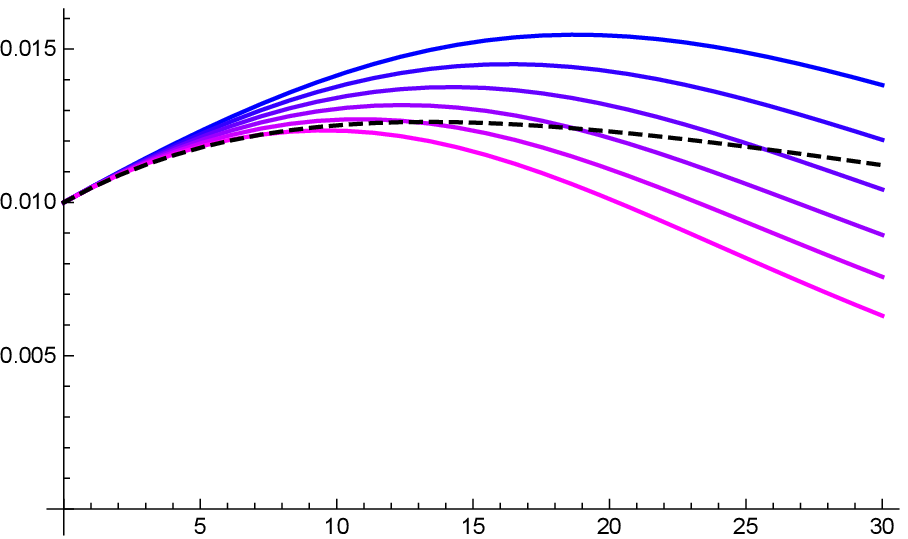}&
\includegraphics[width=0.48\textwidth]{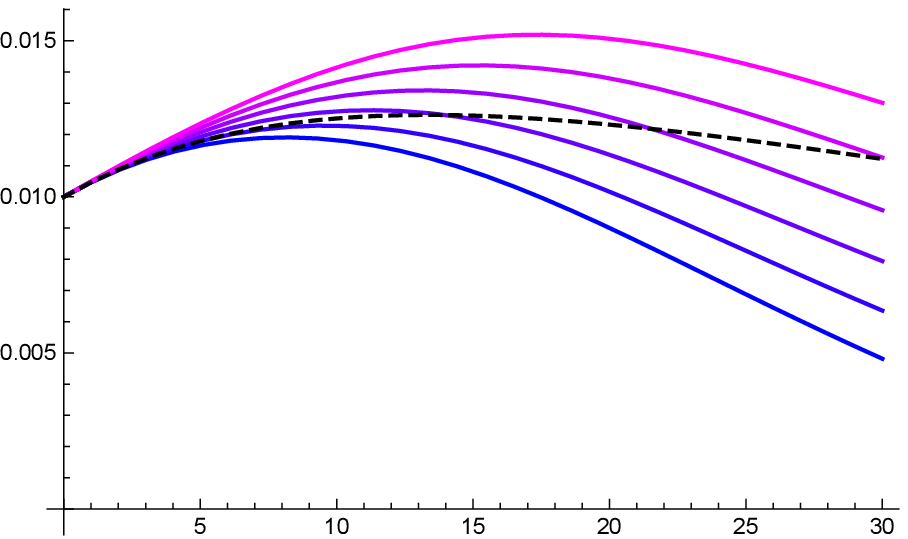}\\
Effect of $\gam$ on $y$ &
Effect of $\nu$ on $y$
\end{tabular}
\caption{
\textit{Left}: Indifference yields $y(t,x,r;T,\nu, \gam)$ as a function of $T$ with $\nu = 1$ fixed and risk aversion varying from $\gam = 0.1$ (\blu{blue}) to $\gam = 0.2$ (\magenta{magenta}).
\textit{Right}: Indifference yields $y(t,x,r;T,\nu, \gam)$ as a function of $T$ with $\gam = 0.15$ and $\nu$ varying from $\nu = -4.5$ (\blu{blue}) to $\nu = 4.5$ (\magenta{magenta}).
\textit{Right and Left}: The dashed black line if the risk-neutral yield curve $\yt(t,x,r;T,\lam)$ with $\lam = 0$ fixed.
}
\label{fig:gam}
\end{figure}

\begin{figure}
\begin{tabular}{cc}
\includegraphics[width=0.48\textwidth]{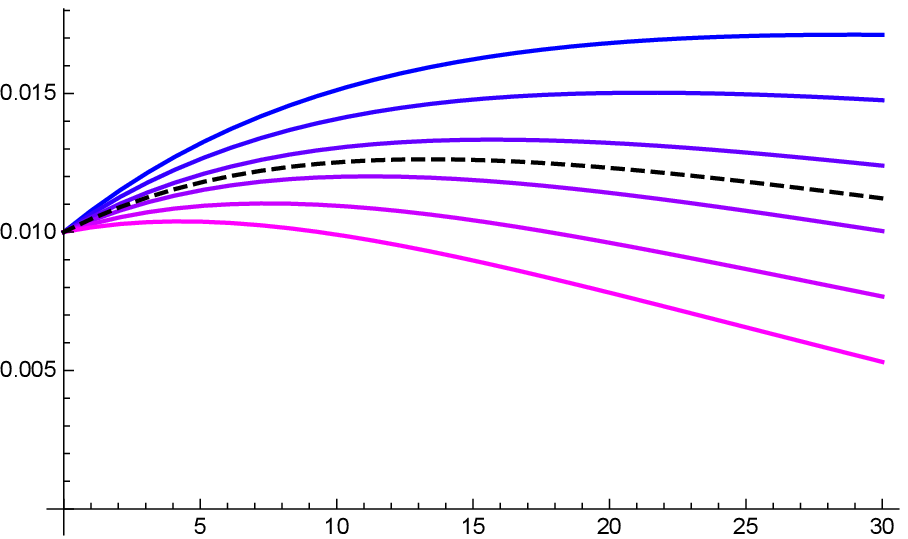}&
\includegraphics[width=0.48\textwidth]{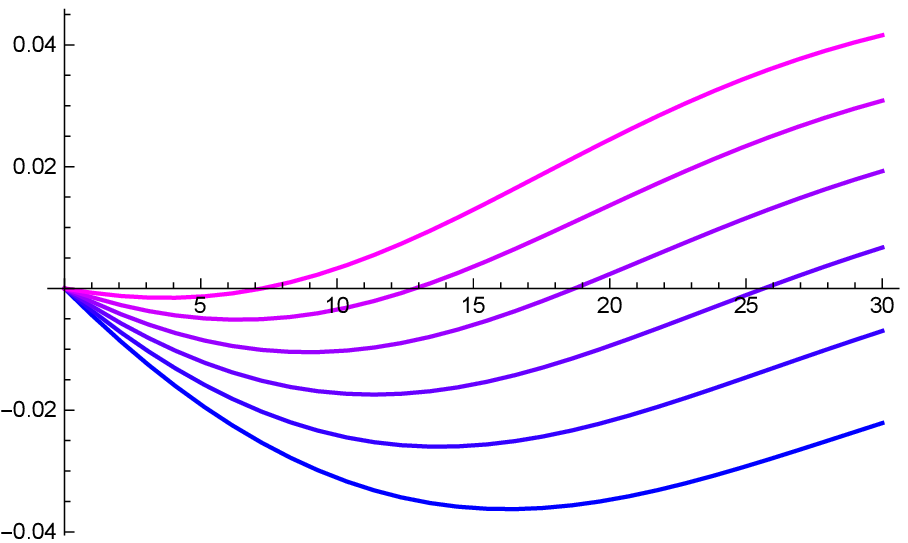}\\
Effect $\lam$ on $\yt$ &
Effect of $\gam$ on $\lam$
\end{tabular}
\caption{
\textit{Left}: Risk-neutral yields $\yt(t,r;T,\lam)$ as a function of $T$ with $\lam$ varying from $\lam = -0.05$ (\blu{blue}) to $\lam = 0.05$ (\magenta{magenta}).
The black-dashed line corresponds to $\lam = 0$.
\textit{Right}: Investor's indifference price of interest rate risk $\lam(t,x,r;T,\nu,\gam)$ as a function of $T$ with $\nu = 1$ fixed and the risk aversion varying from $\gam = 0.1$ (\blu{blue}) to $\gam = 0.2$ (\magenta{magenta}).
}
\label{fig:lambda}
\end{figure}

\end{document}